\documentclass[12pt]{article}

\usepackage{amscd}
\usepackage{amsmath}
\usepackage{amssymb}
\usepackage{amstext}
\usepackage{amsthm}
\usepackage{bbold}
\usepackage{bm}
\usepackage{booktabs}
\usepackage{color}
\usepackage{colonequals}
\usepackage{comment}

\usepackage{easybmat}
\usepackage{etex}
\usepackage{enumitem}
\usepackage{framed}
\usepackage{authblk}
\usepackage[dvips,letterpaper,margin=1in]{geometry}
\usepackage{graphicx}
\usepackage{listings}
\usepackage{longtable}
\usepackage{makecell}
\usepackage{mathtools}
\usepackage{multirow}
\usepackage{rotating}
\usepackage{setspace}
\usepackage[group-separator={,}]{siunitx}
\usepackage{diagbox}
\usepackage{tabu}
\usepackage{verbatim}

\definecolor{cblack}{rgb}{0,0,0}
\definecolor{cblue}{rgb}{0.121569,0.466667,0.705882}    
\definecolor{corange}{rgb}{1.000000,0.498039,0.054902}  
\definecolor{cgreen}{rgb}{0.172549,0.627451,0.172549}   
\definecolor{cred}{rgb}{0.839216,0.152941,0.156863}     
\definecolor{cpurple}{rgb}{0.580392,0.403922,0.741176}  
\definecolor{cbrown}{rgb}{0.549020,0.337255,0.294118}   
\definecolor{cpink}{rgb}{0.890196,0.466667,0.760784}
\definecolor{cgray}{rgb}{0.498039,0.498039,0.498039}
\definecolor{cgreen2}{rgb}{0.7372549019607844, 0.7411764705882353, 0.13333333333333333}

\usepackage{hyperref}
\hypersetup{
  linkcolor  = cblue,
  citecolor  = cgreen,
  urlcolor   = corange,
  colorlinks = true,
}

\newtheorem{theorem}{Theorem}[section]
\newtheorem{remark}[theorem]{Remark}

\newtheorem{lemma}[theorem]{Lemma}

\newtheorem{definition}[theorem]{Definition}

\newtheorem{corollary}[theorem]{Corollary}
\newtheorem{conjecture}[theorem]{Conjecture}
\theoremstyle{plain} 
\newcommand{\thistheoremname}{}
\newtheorem*{genericthm}{\thistheoremname}

\newcommand{\Poincare}{Poincar\'{e}}

\newcommand{\what}{\widehat}

\makeatletter
\def\moverlay{\mathpalette\mov@rlay}
\def\mov@rlay#1#2{\leavevmode\vtop{%
   \baselineskip\z@skip \lineskiplimit-\maxdimen
   \ialign{\hfil$\m@th#1##$\hfil\cr#2\crcr}}}
\newcommand{\charfusion}[3][\mathord]{
    #1{\ifx#1\mathop\vphantom{#2}\fi
        \mathpalette\mov@rlay{#2\cr#3}
      }
    \ifx#1\mathop\expandafter\displaylimits\fi}
\makeatother

\newcommand{\EE}{\mathbb{E}}

\newcommand{\NN}{\mathbb{N}}
\newcommand{\PP}{\mathbb{P}}
\newcommand{\QQ}{\mathbb{Q}}
\newcommand{\RR}{\mathbb{R}}
\renewcommand{\SS}{\mathbb{S}}

\DeclareSymbolFont{bbold}{U}{bbold}{m}{n}
\DeclareSymbolFontAlphabet{\mathbbold}{bbold}
\newcommand{\One}{\mathbbold{1}}

\newcommand{\bx}{\bm x}
\newcommand{\by}{\bm y}
\newcommand{\bz}{\bm z}
\newcommand{\bA}{\bm A}
\newcommand{\bB}{\bm B}

\newcommand{\bJ}{\bm J}

\newcommand{\bP}{\bm P}
\newcommand{\bQ}{\bm Q}

\newcommand{\bU}{\bm U}
\newcommand{\bV}{\bm V}
\newcommand{\bW}{\bm W}

\newcommand{\one}{\bm{1}}

\newcommand{\sN}{\mathcal{N}}
\newcommand{\sO}{\mathcal{O}}

\DeclareSymbolFont{sfoperators}{OT1}{cmss}{m}{n}
\DeclareSymbolFontAlphabet{\mathsf}{sfoperators}
\makeatletter
\renewcommand{\operator@font}{\mathgroup\symsfoperators}
\makeatother

\DeclareMathOperator{\sym}{sym}

\DeclareMathOperator{\Tr}{Tr}
\DeclareMathOperator{\rank}{rank}

\DeclareMathOperator{\Unif}{Unif}

\newcommand{\Ex}{\mathop{\mathbb{E}}}  

\newcommand{\Prx}{\mathop{\mathrm{Pr}}}

\newcommand{\Haar}{\mathsf{Haar}}

\newcommand{\TV}{\mathsf{TV}}

\newcommand{\samp}{\mathtt{samp}}
\newcommand{\aapprox}{\mathtt{approx}}

\newcommand{\proj}{\mathsf{proj}}

\renewcommand{\epsilon}{\varepsilon}

\newcommand\numberthis{\addtocounter{equation}{1}\tag{\theequation}}

\title{Optimality of Glauber dynamics for general-purpose Ising model sampling and free energy approximation}
\date{November 30, 2023}

\usepackage{authblk}
\author{Dmitriy Kunisky\thanks{Email: \texttt{dmitriy.kunisky@yale.edu}. Partially supported by ONR Award N00014-20-1-2335 and a Simons Investigator Award to Daniel Spielman.}}
\affil{Department of Computer Science, Yale University}

\begin{document}

\maketitle

\thispagestyle{empty}

\begin{abstract}
    Recently, Eldan, Koehler, and Zeitouni (2020) showed that Glauber dynamics mixes rapidly for general Ising models so long as the difference between the largest and smallest eigenvalues of the coupling matrix is at most $1 - \epsilon$ for any fixed $\epsilon > 0$.
    We give evidence that Glauber dynamics is in fact optimal for this ``general-purpose sampling'' task.
    Namely, we give an average-case reduction from hypothesis testing in a Wishart negatively-spiked matrix model to approximately sampling from the Gibbs measure of a general Ising model for which the difference between the largest and smallest eigenvalues of the coupling matrix is at most $1 + \epsilon$ for any fixed $\epsilon > 0$.
    Combined with results of Bandeira, Kunisky, and Wein (2019) that analyze low-degree polynomial algorithms to give evidence for the hardness of the former spiked matrix problem, our results in turn give evidence for the hardness of general-purpose sampling improving on Glauber dynamics.
    We also give a similar reduction to approximating the free energy of general Ising models, and again infer evidence that simulated annealing algorithms based on Glauber dynamics are optimal in the general-purpose setting.
\end{abstract}

\clearpage

\setcounter{page}{1}
\pagestyle{plain}

\section{Introduction}

We study algorithmic questions involving the Gibbs measure of an Ising model with coupling matrix $\bJ \in \RR^{n \times n}_{\sym}$, which is the probability measure on $\{ \pm 1 \}^n$ having weights
\begin{equation}
    \label{eq:mu}
    \mu(\bx) = \mu_{\bJ}(\bx) \colonequals \frac{1}{Z}\exp\left(\frac{1}{2}\bx^{\top} \bJ \bx\right).
\end{equation}
Here, the \emph{partition function} $Z = Z(\bJ)$ is the quantity
\begin{equation}
    Z = Z(\bJ) \colonequals \sum_{\bx \in \{\pm 1\}^n} \exp\left(\frac{1}{2}\bx^{\top} \bJ \bx \right).
\end{equation}
We focus on two standard tasks associated to such a model: (1) approximately sampling from $\mu_{\bJ}$, and (2) approximating the partition function or, equivalently, the \emph{free energy} $\log Z(\bJ)$.

Existing algorithmic results for such problems may be separated into two categories.
On the one hand, one may try to produce algorithms for working with $\mu_{\bJ}$ for broad classes of $\bJ$, subject only to modest constraints on ``conditioning'' quantities like the range of eigenvalues or a measurement of the effective rank of $\bJ$.
This is the kind of algorithm we will study, which we call a \emph{general-purpose} algorithm.
It is worth mentioning, however, that one may on the other hand study either specific deterministic sequences of $\bJ$ or specific probability distributions of $\bJ$ (say, the adjacency matrix of a random graph, or the Gaussian orthogonal ensemble giving rise to the Sherrington-Kirkpatrick model).
In this case, one may tailor algorithms to the $\bJ$ in question and achieve stronger performance than general-purpose algorithms; see Section~\ref{sec:related} for examples.

In recent years, several interesting results have appeared on general-purpose algorithms for Ising models under spectral conditions on $\bJ$.
These include the following result on sampling from $\mu_{\bJ}$.

\begin{definition}
    A \emph{sampler} is a randomized algorithm that takes as input $\bJ \in \RR^{n \times n}_{\sym}$ and outputs a random $\bx \in \{\pm 1\}^n$.
    If we call $\samp$ such an algorithm, we write $\samp(\bJ)$ for the law of this random variable $\bx$.
\end{definition}

\begin{theorem}[\cite{EKZ-2020-SpectralConditionFastMixing,AJKPV-2021-EntropicIndependenceI}]
    \label{thm:glauber}
    Let $\delta, \epsilon > 0$.
    There is a sampler $\samp$ running in time polynomial in $n$, $\log(\delta^{-1})$, and $\epsilon^{-1}$ such that, for any $\bJ$ with $\lambda_{\max}(\bJ) - \lambda_{\min}(\bJ) \leq 1 - \epsilon$, we have $\TV(\samp(\bJ), \mu_{\bJ}) < \delta$.
\end{theorem}
\noindent
Here, $\TV(\cdot\,, \cdot)$ is the total variation distance.
The specific sampling algorithm here is just the familiar (discrete-time) \emph{Glauber dynamics}.\footnote{We do not define these dynamics and some basic related terminology here; see~\cite{LP-2017-MarkovChainsMixingTimes} for exposition.}
An algorithm with this kind of performance, without taking into account the dependence on $\epsilon$, is essentially the same as what is called a \emph{fully polynomial almost uniform sampler (FPAUS)} for combinatorial problems, except that the ``target measure'' is not the uniform measure over a complicated set (e.g., matchings or colorings in a graph), but rather a complicated measure $\mu_{\bJ}$ over the simple set $\{\pm 1\}^n$.

The well-known reduction ``from integrating to sampling'' allows us to use an efficient sampler like this to also approximate the free energy.
(For us, ``integrating'' means approximating the partition function or free energy; the reduction ``from counting to sampling'' appearing in the combinatorial setting is also very similar.)
\begin{definition}
    An \emph{approximator} is a randomized algorithm that takes as input $\bJ \in \RR^{n \times n}_{\sym}$ and outputs a random real number.
    If we call $\aapprox$ such an algorithm, we write $\aapprox(\bJ)$ for the law of this random variable.
\end{definition}
\begin{corollary}
    \label{cor:free-energy-estimate}
    Let $\delta, \epsilon > 0$.
    There is an approximator $\aapprox$ computable in polynomial time in $n$, $\delta^{-1}$, and $\epsilon^{-1}$ such that, for any $\bJ$ with $\lambda_{\max}(\bJ) - \lambda_{\min}(\bJ) \leq 1 - \epsilon$, we have with probability $1 - o(1)$ when $F \sim \aapprox(\bJ)$ that
    \begin{equation}
        \label{eq:approximator-guarantee}
        \left|\, \log Z(\bJ) - F \, \right| < \delta.
    \end{equation}
\end{corollary}
\noindent
This algorithm uses the sampler from Theorem~\ref{thm:glauber} together with \emph{simulated annealing} to approximate the free energy.
See Lemma C.4 of~\cite{KLR-2022-SamplingLowRankIsing} for this argument in our Ising model setting, as well as~\cite{Jerrum-2003-CountingSamplingIntegrating,JSV-2004-PolynomialApproximationPermanent,BSVV-2008-AcceleratingSimulatedAnnealingCounting,SVV-2009-AdaptiveSamplingCounting} for similar arguments.\footnote{Usually these reductions produce an approximator satisfying \eqref{eq:approximator-guarantee} with some fixed probability greater than $\frac{1}{2}$, but by the standard trick of taking the median of the outputs of a number of independent runs growing with $n$ we may boost this to probability $1 - o(1)$.}
Such an algorithm, again not taking into account the role of $\epsilon$, is called a \emph{fully polynomial randomized approximation scheme (FPRAS)} for the free energy; quite general notions of equivalence between FPRAS for counting problems or partition functions and FPAUS for associated probability distributions are discussed in the above references.

While the special constant 1 appearing in the conditions on $\bJ$ above that read
\begin{equation}
    \text{`` } \lambda_{\max}(\bJ) - \lambda_{\min}(\bJ) \leq 1 - \epsilon \text{ ''}
\end{equation}
appears naturally in the calculations of~\cite{EKZ-2020-SpectralConditionFastMixing} as well as the similar prior work of~\cite{BB-2019-VerySimpleLSI}, and is optimal for Glauber dynamics according to analysis of the mean-field Curie-Weiss model~\cite{GWL-1966-RelaxationTimeCurieWeiss,LLP-2010-GlauberDynamicsCurieWeiss}, we are not aware of any formal evidence that 1 cannot be replaced with a larger constant in these conditions for a different sampling algorithm.
The goal of this paper is to give evidence that 1 is in fact the optimal constant in such a condition for general-purpose sampling or general-purpose free energy approximation, for arbitrary efficient samplers and approximators.

We will relate samplers or approximators that improve on this constant to hypothesis tests between the following pair of distributions.

\begin{definition}[Wishart spiked matrix model~\cite{Johnstone-2001-LargestEigenvaluePCA}]
    Let $\beta \in [-1, \infty)$, $\gamma > 0$, and, for $n \in \NN$, set $N = N(n) \colonequals \lceil n / \gamma \rceil$.
    We define two probability measures over $(\RR^n)^N$ by the following sampling procedures:
    \begin{enumerate}
    \item Under $\QQ$, draw $\by_1, \dots, \by_N \sim \sN(\bm 0, \bm I_n)$ independently.
    \item Under $\PP$, first draw $\bx \sim \Unif(\{\pm 1 \}^n)$, then draw $\by_1, \dots, \by_N \sim \sN(\bm 0, \bm I_n + \frac{\beta}{n} \bx\bx^{\top})$ independently.
    \end{enumerate}
    Taken together, $\PP$ and $\QQ$ form the \emph{spiked Wishart model} $(\PP, \QQ) \equalscolon \mathsf{Wishart}(n, \beta, \gamma)$, and we write $\mathsf{Wishart}(\beta, \gamma)$ for the sequence of these pairs over $n \in \NN$.
\end{definition}
\noindent
One usually allows the ``prior'' law $\Unif(\{\pm 1\}^n)$ to vary, but we will only be interested in this specific law for our purposes.

The following is the key conjecture based on which we will surmise that our reductions give evidence of hardness of sampling and free energy approximation.
\begin{conjecture}[Hardness of Wishart model~\cite{BKW-2019-ConstrainedPCA}]
    \label{conj:wishart}
    If $\beta > -1$ and $\beta^2 < \gamma$, then there is no polynomial time hypothesis test between the distributions of $\mathsf{Wishart}(\beta, \gamma)$.
    That is, there exists no randomized algorithm $f: (\RR^n)^N \to \{\mathtt{p}, \mathtt{q}\}$ that runs in polynomial time in $n$ and that has
    \begin{equation}
        \lim_{n \to \infty} \PP[f(\by_1, \dots, \by_N) = \mathtt{p}] = \lim_{n \to \infty} \QQ[f(\by_1, \dots, \by_N) = \mathtt{q}\hspace{0.1em}] = 1.
    \end{equation}
\end{conjecture}
\noindent
The main evidence for this conjecture at this level of generality comes from the analysis of \emph{low-degree polynomial algorithms}.
That specific evidence is given in~\cite{BKW-2019-ConstrainedPCA}; see also~\cite{Hopkins-2018-Thesis, KWB-2022-LowDegreeNotes} for general surveys and~\cite{MW-2021-PlantedVectorSubspace} for similar analysis of other priors.
Roughly speaking,~\cite{BKW-2019-ConstrainedPCA} show that no polynomial of degree $o(n / \log(n))$ can successfully perform hypothesis testing in the setting of Conjecture~\ref{conj:wishart} (albeit when the high-probability notion of success above is replaced with a softer condition involving expectations and variances of a polynomial).

One may also view the older analysis of the \emph{Baik--Ben Arous--P\'{e}ch\'{e}~(BBP) transition}~\cite{BBAP-2005-LargestEigenvalueSampleCovariance} in random matrix theory (observed empirically by Johnstone in the same work that introduced the spiked matrix model~\cite{Johnstone-2001-LargestEigenvaluePCA}) as some evidence of hardness.
This result shows that, if $\beta > 0$, then once $\beta^2 < \gamma$ the specific natural test statistic $\lambda_{\max}(\frac{1}{N}\sum_{i = 1}^N \by_i\by_i^{\top})$, the largest eigenvalue of the sample covariance, fails to distinguish $\PP$ from $\QQ$ (in the above high-probability sense).
For the \emph{negatively-spiked} case $\beta < 0$, the natural test statistic is instead $\lambda_{\min}(\frac{1}{N}\sum_{i = 1}^N \by_i\by_i^{\top})$, and again when $\beta^2 < \gamma$ then this fails to distinguish $\PP$ from~$\QQ$~\cite{BS-2006-EigenvaluesSampleCovariance}.
Our reduction will rely on the latter negatively-spiked case, and we will work only with $\gamma > 1$, which always falls in the computationally-hard regime since we will have $\beta \in (-1, 0)$, as we must for the negatively-spiked model to be defined at all.

Finally, note that the assumption $\beta > -1$ is necessary; otherwise, the $\by_i$ are exactly orthogonal to $\bx$ under $\PP$, in which case lattice-based algorithms can recover $\bx$~\cite{ZSWB-2022-LatticeMethodsSOS,DK-2022-ComponentAnalysisLatticeBasis}.
Our reduction will work in this regime, so this nuance will not cause difficulties.

\subsection{Main Results}

We are now prepared to state our main result.
\begin{theorem}
    \label{thm:free-energy}
    For any $\epsilon > 0$, there exists $\delta > 0$ such that the following holds.
    Suppose that there is an approximator $\aapprox$ that runs in polynomial time in $n$ and such that, for any $\bJ$ with $\lambda_{\max}(\bJ) - \lambda_{\min}(\bJ) \leq 1 + \epsilon$, we have with probability $1 - o(1)$ that, when $F \sim \aapprox(\bJ)$,
    \begin{equation}
        \left| \, \log Z(\bJ) - F \, \right| < \delta n.
    \end{equation}
    Then, Conjecture~\ref{conj:wishart} is false.
\end{theorem}
\noindent
We note that the factor of $n$ on the right-hand side means that our lower bound holds against approximators that are allowed to make a significantly larger error in approximating the free energy than the kind of approximation demanded by a FPRAS as in Corollary~\ref{cor:free-energy-estimate}.

The following analogous statement for sampling is then immediate from the Theorem by the same reduction from integration to sampling as discussed above for Corollary~\ref{cor:free-energy-estimate}.

\begin{corollary}
    \label{cor:sampling}
    Let $\epsilon > 0$.
    Suppose that, for any $\delta > 0$, there is a sampler $\samp$ that runs in polynomial time in $n$ and $\log(\delta^{-1})$ and such that, for any $\bJ$ with $\lambda_{\max}(\bJ) - \lambda_{\min}(\bJ) \leq 1 + \epsilon$, we have $\TV(\samp(\bJ), \mu_{\bJ}) < \delta$.
    Then, Conjecture~\ref{conj:wishart} is false.
\end{corollary}
\noindent
In short, the proof is that if such a sampler exists, then by the same simulated annealing strategy cited earlier, we may use this sampler to produce an approximator satisfying the conditions of Theorem~\ref{thm:free-energy}.

Taken together, these results give the evidence we have promised: conditional on Conjecture~\ref{conj:wishart}, the constant 1 appearing in the conditions of both Theorem~\ref{thm:glauber} and Corollary~\ref{cor:free-energy-estimate} cannot be improved.

\subsection{Related Work}
\label{sec:related}

\paragraph{General-purpose algorithms for and analysis of Ising models}
There are numerous conditions in the literature for rapid mixing of Glauber dynamics for Ising models as well as other regularity conditions on the Gibbs measure.
One of the most widely studied is \emph{Dobrushin's uniqueness condition} \cite{Dobrushin-1968-RandomFieldRegularity}, which in this context asks that $\sum_{j = 1}^n |J_{ij}| < 1$ for all $i$.
This addresses some models of interest, but importantly does not treat dense spin glass models (in the temperature regimes where they undergo interesting phase transitions), where $J_{ij}$ is random and of typical order $\Theta(n^{-1/2})$, such as the much-studied Sherrington-Kirkpatrick~(SK) model \cite{SK-1975-SolvableModel, Parisi-1979-SK, Parisi-1980-SequenceApproxSK, Talagrand-2006-Parisi, Panchenko-2013-SK}.
Recently, \cite{BB-2019-VerySimpleLSI} proved a log-Sobolev inequality for $\mu_{\bJ}$ under a bound only on the difference between the largest and smallest eigenvalues of $\bJ$, a breakthrough result since it was the first of its kind to apply to the SK model.\footnote{Note that $\mu_{\bJ} = \mu_{\bJ + c\bm I}$ for any $c$. For this reason, some of the results we cite assume that $\bJ \succeq \bm 0$ and refer to $\|\bJ\|$ instead of $\lambda_{\max}(\bJ) - \lambda_{\min}(\bJ)$, but this is without loss of generality, and we do not make this transformation to avoid confusion about the generality of both our results and others'.}
As discussed by \cite{EKZ-2020-SpectralConditionFastMixing}, this still does not ensure rapid mixing of Glauber dynamics; to that end, they proved a suitable \Poincare\ inequality for $\mu_{\bJ}$ under the same condition.
Since then, several follow-up works have elaborated on the algorithmic consequences; we highlight the result of \cite{KLR-2022-SamplingLowRankIsing} which gives modified algorithms allowing for a small number of outlier eigenvalues in $\bJ$, and those of \cite{AJKPV-2021-EntropicIndependenceI}, who treat a more general class of random walks related to high-dimensional expanders and, by proving modified log-Sobolev inequalities for the associated measures, also sharpen the mixing time bound obtained by \cite{EKZ-2020-SpectralConditionFastMixing} for Glauber dynamics on Ising models.

\paragraph{Computational complexity of sampling}
A notable line of work \cite{Sly-2010-ComputationalTransitionUniquenessThreshold,SS-2012-HardnessPartitionFunction2Spin, GSV-2015-InapproximabilityAntiferromagneticTree,GSV-2016-HardnessPartitionAntiferromagneticIsing,GSVY-2016-FerromagneticPottsHardness} gave inapproximability results for partition functions of Ising, hard-core, and Potts models on $d$-regular graphs with temperature below a relevant transition point for the infinite $d$-ary tree.
For our setting, translating the notation of \cite{GSV-2016-HardnessPartitionAntiferromagneticIsing}, the anti-ferromagnetic Ising model corresponds to an Ising model in our notation where $\bJ = -\eta \bA$ for $\bA$ the adjacency matrix of a $d$-regular graph and $\eta > 0$.\footnote{In such a setting, $\eta^{-1}$ is referred to as the \emph{temperature}. To avoid confusion with the $\beta$ appearing in the Wishart spiked matrix model, we use $\eta$ for inverse temperatures.}
Their main result is that, for any fixed $d \geq 3$ and fixed $\eta > \frac{1}{2}\log(\frac{d}{d - 2}) = \frac{1}{2}\log(1 + \frac{2}{d - 2})$, unless $\mathsf{RP} = \mathsf{NP}$, there is no algorithm approximating the partition function of this model on a given input graph.\footnote{An anti-ferromagnetic Ising model on a graph $G = (V, E)$ whose Gibbs measure gives $\bx \in \{\pm 1\}^V$ a weight of $B^{\#\{\{i, j\} \in E: x_i = x_j\}}$ with $B < 1$ is equivalent to an Ising model in our notation with $\bJ = -\frac{\log(1/B)}{2} \bA$ for $\bA$ the adjacency matrix of $G$. $B = \frac{d - 2}{d}$ is the critical parameter associated to the infinite $d$-ary tree.}
Since $-d \leq \lambda_{\min}(\bA) < \lambda_{\max}(\bA) = d$ for any $d$-regular $\bA$ (with the inequality achieved by bipartite graphs), this gives us a stronger hardness result that general-purpose sampling is impossible over $\bJ$ with $\lambda_{\max}(\bJ) - \lambda_{\min}(\bJ) \leq d\log(1 + \frac{2}{d - 2}) + \epsilon$ for any $\epsilon > 0$.
This function of $d$ decreases monotonically to a limit of 2 as $d \to \infty$.
This implies, for instance, that Corollary~\ref{cor:sampling} holds with the constant 1 replaced by 2, but with the refutation of Conjecture~\ref{conj:wishart} replaced by the much stronger conclusion that $\mathsf{RP} = \mathsf{NP}$.\footnote{A version of Theorem~\ref{thm:free-energy} also follows, but only holds against approximators with stronger guarantees.}
We offer the following conjecture, which asserts that we can have the ``best of both worlds.''
\begin{conjecture}
    Theorem~\ref{thm:glauber} and Corollary~\ref{cor:sampling} hold with the refutation of Conjecture~\ref{conj:wishart} replaced by the conclusion $\mathsf{RP} = \mathsf{NP}$.
\end{conjecture}
\noindent
We note that the results cited above depend deeply on the graph setting, while our argument uses quite different coupling matrices more closely related to the theory of dense spin glasses, so this Conjecture seems to require substantially novel techniques.

Another lower bound quite similar to ours, based on a reduction to sampling from a suitable random approximation problem, is that of \cite{AA-2011-ComplexityLinearOptics}, though it pertains to a very different setting motivated by quantum computing.
As in our case, that result (their Theorem~3) appeals to an underlying hardness conjecture, but allows the conclusion of strong hardness results against arbitrary sampling algorithms.

\paragraph{Lower bounds against specific samplers}
Most arguments for hardness of sampling that we are aware of concern Markov Chain Monte Carlo (MCMC) methods and study only a specific choice of dynamics (e.g., Glauber dynamics in our setting).
Then, it is possible to show slow mixing by analyzing geometric properties such as ``bottlenecks'' in the configuration space that the Markov chain navigates.
See, e.g., Chapter~7 and Section~15.6 of \cite{LP-2017-MarkovChainsMixingTimes} for some examples.
In particular, such results are known for Glauber dynamics run on the \emph{Curie-Weiss model}, which in our notation is an Ising model with $\bJ = \frac{\eta}{n}\one\one^{\top}$ for some $\eta > 0$.
In this case, it is known that Glauber dynamics mixes rapidly if $\eta < 1$ and slowly if $\eta > 1$ \cite{GWL-1966-RelaxationTimeCurieWeiss,LLP-2010-GlauberDynamicsCurieWeiss}, which is precisely compatible with our result since $\lambda_{\max}(\bJ) - \lambda_{\min}(\bJ) = \eta$.
Our result additionally gives evidence that no \emph{other} sampler can improve on this threshold.
Another interesting idea, albeit one relying heavily on the special spin glass phenomenon of \emph{disorder chaos}, is used by \cite{EAMS-2022-SamplingSKStochasticLocalization} to rule out a broad class of ``stable'' sampling algorithms for the specific case of the SK model at sufficiently low temperature.

\paragraph{Sherrington-Kirkpatrick (SK) model}
The SK model has $\bJ = \eta \bW$ for $\eta > 0$ and $\bW$ drawn from the Gaussian orthogonal ensemble (GOE), i.e., with independent entries on and above the diagonal drawn as $W_{ij} = W_{ji} \sim \sN(0, (1 + \One\{i = j\})n^{-1})$.
Per classical random matrix theory \cite{AGZ-2010-RandomMatrices}, such $\bW$ has $\lambda_{\max}(\bW)  = 2 + o(1)$ and $\lambda_{\min}(\bW) = -2 - o(1)$ with high probability, so the results of \cite{EKZ-2020-SpectralConditionFastMixing,AJKPV-2021-EntropicIndependenceI} imply rapid mixing of Glauber dynamics for $\eta < \frac{1}{4}$.
Meanwhile, it was long conjectured in the physics literature that Glauber dynamics should mix rapidly for all $\eta < 1$, the threshold of \emph{replica symmetry breaking (RSB)}, a geometric property of $\mu_{\bJ}$ expected to prevent rapid mixing \cite{SZ-1981-DynamicTheorySpinGlass,MPV-1987-SpinGlassTheoryBeyond}.
By designing a sampler tailored to the SK model, \cite{EAMS-2022-SamplingSKStochasticLocalization} showed that it is possible to sample efficiently for $\eta < \frac{1}{2}$, and, soon after, \cite{Celentano-2022-SudakovFerniquePostAMP} closed the remaining gap and showed that the same algorithm succeeds for all $\eta < 1$.\footnote{However, the results of \cite{EAMS-2022-SamplingSKStochasticLocalization,Celentano-2022-SudakovFerniquePostAMP} only give sampling guarantees in Wasserstein distance with respect to the $\ell^2$ metric, a weaker metric than the total variation distance appearing in Theorem~\ref{thm:glauber}.}

We emphasize the somewhat subtle state of affairs: there is good reason to believe that Glauber dynamics, not just the more complicated sampler of \cite{EAMS-2022-SamplingSKStochasticLocalization}, mixes rapidly for the SK model for all $\eta < 1$.
However, as we show in this paper, there is also reason to believe that this is \emph{not} because the spectral condition of \cite{EKZ-2020-SpectralConditionFastMixing} is suboptimal, and thus showing this rapid mixing would require analysis relying on further properties specific to the coupling matrix $\bJ$ of the SK model.

\paragraph{Spin glasses and replica symmetry}
Our main argument in Section~\ref{sec:proof} will proceed by comparing two disordered Ising models coming from two different laws for the random coupling matrix $\bJ$: on the one hand we will consider an orthogonally invariant random matrix $\bJ$, in fact just a rescaled projection matrix to a uniformly random low-dimensional subspace, and on the other a rescaled projection whose row space is biased towards a hypercube vector.
These distributions and their connections to low-degree polynomial algorithms are borrowed from the prior work of \cite{BKW-2019-ConstrainedPCA}; see the discussion below.
Our analysis of the former distribution is related to a long line of work in probability theory and statistical physics on the behavior of Ising models with general orthogonally invariant couplings  \cite{MPR-1994-NonRandomSpinGlassROM,PP-1995-TAPEquationsOrthogonalIsing,OW-2001-TAPEquationsOrthogonalIsing,CDL-2003-OrthogonallyInvariantIsing,BS-2016-HighTemperatureOrthogonalSpinGlass,MFCKMZ-2019-PlefkaExpansionOrthogonalIsing,FW-2021-ReplicaSymmetricInvariantCouplings,FLS-2022-TAPEquationsOrthogonalIsing}.
The specific $\bJ$ we consider is essentially equivalent to the ``random orthogonal model (ROM)'' of \cite{MPR-1994-NonRandomSpinGlassROM}, which considers $\bJ$ proportional to $2\bP - \bm I$ for $\bP$ an orthogonal projection matrix.

If we consider our model with $\rank(\bP) = \delta n$, then for each $\delta$ there is a threshold $\eta_{\,\mathsf{RSB}}(\delta)$ such that for $\eta > \eta_{\,\mathsf{RSB}}(\delta)$ the above model is conjectured to exhibit RSB \cite{SZ-1981-DynamicTheorySpinGlass,MPV-1987-SpinGlassTheoryBeyond}, and we have $\eta_{\,\mathsf{RSB}}(\delta) \to 1$ as $\delta \to 0$.\footnote{The former claim is discussed in the above references, and the latter may be obtained by following the analysis of the de~Almeida--Thouless~(AT) condition for stability of the replica-symmetric solution for these models from \cite{OW-2001-TAPEquationsOrthogonalIsing}.}
This gives some explanation of the ``magic constant'' 1 appearing in our results: the hardness of general-purpose sampling is witnessed by RSB in a particular distribution of coupling matrices, the orthogonally invariant low-rank projections.

However, the mathematically rigorous results among those above only concern high temperature, i.e., $\bJ = \eta \bP$ for $\eta$ sufficiently small, and in particular only aim to describe a portion of the regime $\eta < \eta_{\mathsf{RSB}}$, while we will be interested in taking $\eta > 1$ (by a small margin), so that the model is just \emph{over} the threshold of RSB.
Analysis in this regime seems quite challenging, and in any case the relationship between RSB and hardness of sampling is only conjectural (though established in some special cases, e.g., by \cite{EAMS-2022-SamplingSKStochasticLocalization} for stable algorithms for the SK model), so we instead proceed by reducing to a hypothesis testing problem for which it is easier to find evidence of hardness.
For this, we will only need one-sided bounds on the partition function or free energy, for which we apply a result of \cite{Comets-1996-SKSphericalBound} on comparing Ising models to \emph{spherical} models with the same coupling matrix, which are much simpler to analyze.
At high temperature, these models have the same limiting free energy (as shown formally by  \cite{BS-2016-HighTemperatureOrthogonalSpinGlass,FW-2021-ReplicaSymmetricInvariantCouplings}), but at low temperature the spherical model's free energy still bounds the Ising model's.
In our setting, these facts are identical to the equivalence of ``quenched'' and ``annealed'' calculations at high temperature, and annealed calculations still bounding quenched ones at low temperature.

\paragraph{Low-degree polynomial algorithms}
The analysis of low-degree polynomials as algorithms for hypothesis testing problems was inspired by the line of work leading to lower bounds for the planted clique problem in the sum-of-squares hierarchy; see in particular \cite{HKPRS-2018-PlantedCliqueSOS4,BHKKMP-2019-PlantedClique} for this line of work as well as \cite{HS-2017-BayesianEstimation,HKPRSS-2017-SOSSpectral,Hopkins-2018-Thesis,KWB-2022-LowDegreeNotes} for further exposition and applications.
More recently, this has grown into a convenient, reliable, and broadly applicable technique for proving lower bounds for average-case problems of many kinds \cite{GJW-2020-LowDegreeOptimization,Kunisky-2020-LowDegreeMorris,Wein-2020-LowDegreeIndependentSet,SW-2020-LowDegreeEstimation,Wein-2022-LowDegreeTensorDecomposition,RSWY-2022-CountCommunitiesLowDegree,KVWX-2023-LowDegreeColoringClique,MWZ-2023-DetectionRecoveryPlantedCycles}.

Especially influential on our approach is the technique of \cite{BKW-2019-ConstrainedPCA} (later also used in \cite{BBKMW-2020-SpectralPlantingColoring,BKW-2020-PositivePCA}), who first showed a reduction from the certification of an upper bound\footnote{This means computing a function of $\bW$ that is an upper bound on $\mathsf{M}(\bW)$ for all $\bW$.} on $\mathsf{M}(\bW) \colonequals \frac{1}{n}\max_{\bx \in \{\pm 1\}^n} \bx^{\top}\bW\bx$ for $\bW$ as in the SK model to hypothesis testing, and then analyzed that hypothesis testing problem using low-degree polynomials.
Indeed, our Conjecture~\ref{conj:wishart} is supported by the second half of that analysis, and the optimization problem $\mathsf{M}(\bW)$ is the ``zero temperature'' version of the free energy problem we study, or the problem of computing the ``ground state energy.''
Softening that argument to apply to the free energy at positive temperature paves the way to further reducing to sampling.

Finally, let us remark that we see our results as being double-edged: on the one hand, to the extent that we believe the general ``low-degree conjecture'' that low-degree polynomial algorithms are optimal in some settings (see \cite{Hopkins-2018-Thesis, KWB-2022-LowDegreeNotes} for various attempts at precise statements), our results indeed give evidence for hardness of sampling and free energy approximation, as we have presented them above.
On the other hand, these reductions describe new algorithms, ones very different from low-degree polynomials, for attempting to solve problems like hypothesis testing in the Wishart spiked matrix model, and therefore perhaps make it more plausible that such algorithms may eventually refute the low-degree conjecture (i.e., show that another type of algorithm can perform substantially better at a natural hypothesis testing task than low-degree polynomials).

\paragraph{Certification, approximation, and sampling} We outline an argument that inspired this work, related to the above case of the SK model and due to Andrea Montanari (reproduced here from private communications with his generous permission).
The specific result of \cite{BKW-2019-ConstrainedPCA} is that it is hard for a certification algorithm to improve on the \emph{spectral bound} $\mathsf{M}(\bW) \leq \|\bW\| \approx 2$ for $\bW$ drawn from the GOE.
On the other hand, we have the pointwise bound, for any $\eta > 0$ and $\bW \in \RR^{n \times n}_{\sym}$, that
\begin{equation}
    \mathsf{M}(\bW) \colonequals \frac{1}{n}\max_{\bx \in \{\pm 1\}^n} \bx^{\top}\bW\bx \leq \frac{1}{\eta \cdot n}\log Z(\eta \bW).
\end{equation}
As mentioned above, for $\bW$ drawn from the GOE as in the SK model, the right-hand side may be efficiently approximated with high probability for all $\eta < 1$ (by the same reduction from integration to sampling invoked earlier).
In fact, when $\eta < 1$, then the expectation of the right-hand side (around which it concentrates) has a simple asymptotic formula, unlike the much more complicated case of $\eta > 1$ (see, e.g., Proposition~2.1 of \cite{ALR-1987-SK}):
\begin{equation}
    \lim_{n \to \infty}\EE \frac{1}{\eta \cdot n}\log Z(\eta \bW) = \frac{\eta}{2} + \frac{2\log(2)}{\eta}.
\end{equation}
For $\eta$ close to but smaller than 1, this is strictly smaller than 2.
If somehow it were possible to convert the sampling-based \emph{approximation} of $\log Z(\eta \bW)$ into a \emph{certifiable bound} for such $\eta$, that would yield an efficient certifiable bound improving on the spectral one, which by the argument of \cite{BKW-2019-ConstrainedPCA} would refute Conjecture~\ref{conj:wishart}.

In other words, conditional on Conjecture~\ref{conj:wishart}, the above gives a hardness result for certifying bounds on the SK model free energy (for certain $\eta$).
Such certification tasks have been investigated by a few works \cite{Risteski-2016-PartitionFunctionsConvexProgramming, RL-2016-MaximumEntropyGoemansWilliamson, JKR-2019-MeanFieldConvexHierarchies}, but are difficult to design algorithms for; in particular, they require a tailored approach since, unlike $\mathsf{M}(\bW)$ itself, the free energy is not directly amenable to algebraic convex relaxations of polynomial optimization problems like the sum-of-squares hierarchy.
This work arose from adapting the above ideas to the perhaps more natural task of sampling or approximating the free energy.
But, in exchange for proving a lower bound against approximation rather than certification, we must expand the class of instances we ask our algorithms to work well on from just typical instances of the SK model to more general coupling matrices.

\section{Notation}

The notation $o(1)$ always refers to the limit $n \to \infty$, and the term ``with high probability'' means with probability $1 - o(1)$.
We write $\SS^{n - 1}(R)$ for the sphere of radius $R$ centered at the origin in $\RR^n$.
We write $\sO(n)$ for the group of $n \times n$ orthogonal matrices, and $\Haar(\sO(n))$ for its Haar measure normalized to be a probability measure.
For a symmetric matrix $\bA \in \RR^{n \times n}_{\sym}$, we write $\lambda_{\max}(\bA) =\lambda_1(\bA) \geq \cdots \geq \lambda_n(\bA) = \lambda_{\min}(\bA)$ for its ordered eigenvalues.
For a further $\bB \in \RR^{n \times n}_{\sym}$, we write $\langle \bA, \bB \rangle \colonequals \Tr(\bA\bB)$.

\section{Preliminaries}

We introduce the following additional thermodynamic quantities associated to the coupling matrix $\bJ$.
The letter ``$p$'' denotes quantities called ``pressures'' in statistical physics parlance (which only differ by a rescaling from free energies).
\begin{align}
  \what{Z}(\bJ) &\colonequals \frac{1}{2^n} Z(\bJ) = \Ex_{\bx \sim \Unif(\{\pm 1\}^n)} \exp\left(\frac{1}{2}\bx^{\top} \bJ \bx \right), \\
  p(\bJ) &\colonequals \frac{1}{n} \log \what{Z}(\bJ).
           \intertext{We also define the following analogous quantities over spherical rather than Ising models, as denoted by the ``$\SS$'' superscript:}
  \what{Z}^{\SS}(\bJ) &\colonequals \Ex_{\bx \sim \Unif(\SS^{n - 1}(\sqrt{n}))} \exp\left(\frac{1}{2}\bx^{\top} \bJ \bx \right), \\
  p^{\SS}(\bJ) &\colonequals \frac{1}{n} \log \what{Z}^{\SS}(\bJ).
\end{align}

We will use two results that we learned of from the work of Comets \cite{Comets-1996-SKSphericalBound}.
We repeat both proofs below for the sake of giving a self-contained exposition.

We call a random matrix $\bJ$ \emph{orthogonally invariant} if $\bJ$ has the same law as $\bQ\bJ\bQ^{\top}$ for any orthogonal matrix $\bQ$.
The first result shows that, for orthogonally invariant coupling matrices, the pressure of the Ising model is dominated in expectation by that of the corresponding spherical model.
As mentioned earlier, this is equivalent to the annealed computation of the free energy bounding the quenched one in orthogonally invariant models.
\begin{lemma}[Equation (2.2) of \cite{Comets-1996-SKSphericalBound}]
    \label{lem:spherical-dom}
    Suppose that $\bJ$ is a random matrix whose law is orthogonally invariant.
    Then,
    \begin{equation}
        \Ex_{\bJ} p(\bJ) \leq \Ex_{\bJ} p^{\SS}(\bJ).
    \end{equation}
\end{lemma}
\begin{proof}
    We apply orthogonal invariance and use Jensen's inequality on an expectation over a Haar-distributed orthogonal matrix:
    \begin{align*}
      \Ex_{\bJ} p(\bJ)
      &= \Ex_{\bJ} \Ex_{\bQ \sim \mathsf{Haar}(\sO(n))} p(\bQ^{\top}\bJ \bQ) \tag{orthogonal invariance} \\
      &= \frac{1}{n} \Ex_{\bJ} \Ex_{\bQ \sim \mathsf{Haar}(\sO(n))} \log \Ex_{\bx \sim \Unif(\{\pm 1\}^n)} \exp\left(\frac{1}{2}\bx^{\top} \bQ^{\top}\bJ\bQ \bx\right) \\
      &\leq \frac{1}{n} \Ex_{\bJ} \log \Ex_{\bQ \sim \mathsf{Haar}(\sO(n))} \Ex_{\bx \sim \Unif(\{\pm 1\}^n)} \exp\left(\frac{1}{2}\bx^{\top} \bQ^{\top}\bJ\bQ \bx\right) \tag{Jensen's inequality} \\
      &= \frac{1}{n} \Ex_{\bJ} \log \Ex_{\bx \sim \Unif(\SS^{n - 1}(\sqrt{n}))} \exp\left(\frac{1}{2}\bx^{\top} \bJ \bx\right) \\
      &= \Ex_{\bJ} p^{\SS}(\bJ), \numberthis
    \end{align*}
    completing the proof.
\end{proof}

The use of this observation is that, for any fixed $\bJ$, $\what{Z}^{\SS}(\bJ)$ (and so also $p^{\SS}(\bJ)$) is a function only of the eigenvalues of $\bJ$: thanks to the rotational invariance of the uniform measure on the sphere,
\begin{equation}
    \what{Z}^{\SS}(\bJ) = \Ex_{\bx \sim \Unif(\SS^{n - 1}(\sqrt{n}))} \exp\left(\frac{1}{2}\bx^{\top} \bJ \bx \right) = \Ex_{\bx \sim \Unif(\SS^{n - 1}(\sqrt{n}))} \exp\left(\frac{1}{2}\sum_{i = 1}^n \lambda_i(\bJ) x_i^2 \right).
\end{equation}

Such quantities, a special case of \emph{spherical} or \emph{Harish-Chandra-Itzykson-Zuber integrals}, have been studied at length asymptotically but have no known closed form in general \cite{GM-2005-RTransformSphericalIntegrals}.
However, we have the following useful and simple bound, which is the second result of Comets.
The right-hand side below is the same as the formula for the limiting pressure or free energy of a spherical model with invariant couplings in the replica-symmetric regime; see, e.g., Appendix~D of \cite{FW-2021-ReplicaSymmetricInvariantCouplings}.
The below version is convenient since it gives a non-asymptotic statement and clarifies that the error term does not depend on the model itself at all.
\begin{lemma}[Equation (2.4) of \cite{Comets-1996-SKSphericalBound}]
    \label{lem:spherical-bound}
    For any $\bJ \in \RR^{n \times n}_{\sym}$,
    \begin{equation}
        p^{\SS}(\bJ) \leq \inf_{s > \lambda_{\max}(\bJ) / 2}\left\{ s - \frac{1 + \log(2)}{2} - \frac{1}{2n} \sum_{i = 1}^n \log\left(s - \frac{\lambda_i(\bJ)}{2}\right) \right\} + o(1),
    \end{equation}
    where we emphasize that $o(1)$ hides an absolute sequence going to zero with $n$, not depending on $\bJ$.
\end{lemma}
\begin{proof}
    Fix $\epsilon \in (0, 1)$ and $s > \lambda_{\max}(\bJ) / 2$.
    By a standard evaluation of the Gaussian integral, we have, writing $\nu$ for Lebesgue measure and $A \colonequals \{\bx: (1 - \epsilon)n \leq \|\bx\|^2 \leq n\}$,
    \begin{align*}
      &\hspace{-1cm} \exp\left(\frac{n}{2}\log(2\pi) - \frac{1}{2}\log \det (2s\bm I - \bJ)\right) \\
      &= \int_{\RR^n} \exp\left(\frac{1}{2}\bx^{\top}\bJ \bx - s\|\bx\|^2\right)d\bx \\
      &\geq \int_{A} \exp\left(\frac{1}{2}\bx^{\top}\bJ \bx - s\|\bx\|^2\right)d\bx \\
      &= \nu(A) \Ex_{\bx \sim \Unif(A)} \exp\left(\frac{1}{2}\bx^{\top}\bJ \bx - s\|\bx\|^2\right) \\
      \intertext{and, since when $\bx \sim \Unif(A)$ the norm and direction of $\bx$ are independent, we have, by writing $\what{\bx} \colonequals \sqrt{n} \cdot \bx / \|\bx\|$ and applying Jensen's inequality with respect to $\|\bx\|$,}
      &= \nu(A) \Ex_{\bx \sim \Unif(A)} \exp\left(\frac{\|\bx\|^2}{n}\left[\frac{1}{2}\what{\bx}^{\top}\bJ \what{\bx} - s n \right]\right) \\
      &\geq \nu(A) \Ex_{\what{\bx} \sim \Unif(\SS^{n - 1}(\sqrt{n}))} \exp\left(\left[\Ex_{\bx \sim \Unif(A)}\frac{\|\bx\|^2}{n}\right] \cdot \left[\frac{1}{2}\what{\bx}^{\top}\bJ \what{\bx} - s n \right]\right)
        \intertext{where, since by our assumption on $s$ we have $\frac{1}{2}\what{\bx}^{\top}\bJ \what{\bx} - s n < 0$,}
      &\geq \nu(A) \Ex_{\what{\bx} \sim \Unif(\SS^{n - 1}(\sqrt{n}))} \exp\left(\frac{1}{2}\what{\bx}^{\top}\bJ \what{\bx} - s n \right) \\
      &= \nu(A) \exp(-sn) \what{Z}^{\SS}(\bJ). \numberthis
    \end{align*}
    The result then follows from evaluating the integral of $\nu(A)$ in spherical coordinates and estimating to obtain
    \begin{align*}
      \nu(A)
      &= \pi^{n/2} \frac{n^{n/2}}{\Gamma(1 + n/2)}\left(1 - (1 - \epsilon)^{n/2}\right) \\
      &\geq \exp\left(n\left[\frac{\log(2\pi)}{2} + \frac{1}{2} - o(1)\right]\right), \numberthis
    \end{align*}
    and rearranging the above.
\end{proof}
Similar ideas relating a normalized partition function over the hypercube to a Gaussian integral that may be computed exactly were also used in the recent work of \cite{KZ-2023-OnlineMatrixDiscrepancy} for a different application in discrepancy theory.

\section{Proof of Theorem~\ref{thm:free-energy}}
\label{sec:proof}

We first construct from draws of $\PP$ and $\QQ$ coupling matrices $\bJ$ that we may substitute into our algorithms to distinguish $\PP$ and $\QQ$.
Let us outline the main idea.
As in the arguments of \cite{BKW-2019-ConstrainedPCA}, we will build \emph{projection} matrices from draws of vectors from $\PP$ and $\QQ$; when built from a draw of $\PP$, the projection matrix will have a hypercube vector lying near its row space, while when built from a draw of $\QQ$, its row space will be uniformly random (of a specified dimension).
Therefore, the Ising free energy for the former $\bJ$ will be larger than that for the latter $\bJ$, and approximating the free energy will let us hypothesis test.

Let $(\PP, \QQ) = \mathsf{Wishart}(\beta, \gamma)$ for $\beta \in (-1, 0), \gamma \in (1, \infty)$ to be chosen later.
We write $N = N(n) \colonequals \lceil n / \gamma \rceil \leq n$ as in the definition of the spiked Wishart model, and we also write $N^{\prime} = N^{\prime}(n) \colonequals n - N(n) = n - \lceil n / \gamma \rceil = \lfloor (1 - \gamma^{-1})n \rfloor$.
For any $1 \leq N \leq n$ and any $\by_1, \dots, \by_N \in \RR^n$, let us write $\proj(\by_1, \dots, \by_N) \in \RR^{n \times n}_{\sym}$ for the orthogonal projection matrix to the orthogonal complement of the span of the $\by_1, \dots, \by_N$.
Let us also write $\proj(\PP)$ and $\proj(\QQ)$ for the law of the matrix given by $\proj$ applied to a draw from $\PP$ and of $\QQ$, respectively.
Note that both of these laws are supported on projection matrices of rank $N^{\prime}$.

We now analyze the pressure of Ising models whose coupling matrices are rescalings of projection matrices drawn from each of these laws.
We will use the binary entropy function,
\begin{equation}
    H(x) \colonequals -x\log(x) - (1 - x) \log(1 - x) \text{ for } x \in [0, 1].
\end{equation}

\begin{lemma}
    \label{lem:PP}
    Let $\bP \sim \proj(\PP)$ and $\eta > 0$.
    Define
    \begin{equation}
        \epsilon = \epsilon(\beta, \gamma) \colonequals \frac{\sqrt{2(1 + \beta)}}{\sqrt{\gamma} - 1}.
    \end{equation}
    Then, with high probability,
    \begin{equation}
        p(\eta \bP) \geq \frac{\eta}{2} - \log(2) + \sup_{x \in [0, 1]} \left\{ H(x) + 2\eta x^2 - 2\eta x\right\} - \epsilon \eta - o(1).
    \end{equation}
\end{lemma}
\begin{proof}
    When $(\by_1, \dots, \by_N) \sim \PP$ and $\bP$ is the projection matrix to the orthogonal complement of their span, if $\bx \in \{\pm 1\}^n$ is the ``planted'' vector chosen in the process of sampling from $\PP$, then by calculations in the proof of \cite[Theorem 3.8]{BKW-2019-ConstrainedPCA} we have
    \begin{equation}
        \label{eq:PP-x-fraction}
        \|\bP \bx\|^2 \geq \left(1 - 2\frac{1 + \beta}{(\sqrt{\gamma} - 1)^2}\right) n = (1 - \epsilon^2)n.
    \end{equation}
    with high probability.
    The rest of the proof will be on the event that \eqref{eq:PP-x-fraction} holds.
    On this event, we group the terms in the sum inside of the pressure according to their ``overlap'' with $\bx$, i.e., the number of entries on which they agree with $\bx$:
    \begin{align*}
      p(\eta \bP)
      &= -\log(2) + \frac{1}{n}\log\left(\sum_{\by \in \{\pm 1\}^n} \exp\left(\frac{\eta}{2} \|\bP\by\|^2 \right)\right) \\
      &\geq -\log(2) + \frac{1}{n}\log\left(\sum_{r = 0}^n \binom{n}{r}\exp\left(\frac{\eta}{2}\, \min_{\by: \|\by\| = \sqrt{n}, \langle \bx, \by \rangle = 2r - n}\|\bP\by\|^2 \right)\right) \numberthis \label{eq:p-intermediate-1}
    \end{align*}
    Let $\bz = \bP\bx / \|\bP\bx\|$, a vector lying in the row space of $\bP$.
    Then, for any $\by$ with $\|\by\| = \sqrt{n}$,
    \begin{align*}
      \|\bP\by\|^2
      &\geq \langle \by, \bz \rangle^2 \\
      &= \langle \bP \bx, \by \rangle^2 / \|\bP\bx\|^2 \\
      &\geq \frac{1}{n}\langle \bP \bx, \by \rangle^2 \\
      &= \frac{1}{n}\left(\langle \bx, \by \rangle - \langle (\bm I - \bP)\bx, \by \rangle\right)^2 \\
      &\geq \frac{1}{n}\left(\langle \bx, \by \rangle^2 - 2|\langle \bx, \by \rangle| \cdot |\langle (\bm I - \bP)\bx, \by \rangle|\right)
        \intertext{and by the Cauchy-Schwarz inequality}
      &\geq \frac{1}{n}\left(\langle \bx, \by \rangle^2 - 2n^{3/2} \|(\bm I - \bP)\bx\|\right) \\
      &\geq \frac{1}{n}\langle \bx, \by \rangle^2 - 2\epsilon n, \numberthis
    \end{align*}
    where at the end we use our assumption that \eqref{eq:PP-x-fraction} holds.
    Thus, continuing from above in~\eqref{eq:p-intermediate-1}, we have
    \begin{align*}
      p(\eta \bP)
      &\geq - \log(2) + \frac{1}{n}\log\left(\sum_{r = 0}^n \binom{n}{r}\exp\left(\frac{\eta}{2} \left(\frac{1}{n}(2r - n)^2 - 2\epsilon n\right) \right)\right) \\
      &= -\log(2) - \epsilon \eta + \frac{1}{n}\log\left(\sum_{r = 0}^n \binom{n}{r}\exp\left(\frac{\eta}{2n} (2r - n)^2 \right)\right)
        \intertext{and applying a standard inequality on the binomial coefficient and an approximation argument,}
      &\geq -\log(2) - \epsilon \eta - o(1) + \frac{1}{n}\log\left(\sum_{r = 0}^n \exp\left(n\left[ H\left(\frac{r}{n}\right) + \frac{\eta}{2}\left(2\frac{r}{n} - 1\right)^2\right]\right)\right) \\
      &\geq - \log(2) - \epsilon \eta - o(1) + \sup_{x \in [0, 1]} \left\{ H(x) + \frac{\eta}{2}(2x - 1)^2\right\} \\
      &= \frac{\eta}{2} - \log(2) + \sup_{x \in [0, 1]} \left\{ H(x) + 2\eta x^2 - 2\eta x\right\} -\epsilon \eta - o(1), \numberthis
    \end{align*}
    completing the proof.
\end{proof}
\noindent
We note that the expression for $\epsilon$ explains why we must work in the regime $\beta \in (-1, 0)$: we must have at the very least $\epsilon < 1$ in order for the argument above not to be vacuous, and when $\gamma$ is close to 1 then we must have $\beta < 0$ in order to have $\epsilon < 1$.

\begin{lemma}
    \label{lem:QQ}
    Let $\eta, \delta > 0$.
    There exists $\gamma = \gamma(\eta, \delta) > 1$ such that, when $\bP \sim \proj(\QQ)$, then with high probability
    \begin{equation}
        p(\eta \bP) \leq \frac{\eta}{2} - \frac{1}{2}\log(\eta) - \frac{1}{2} + \delta.
    \end{equation}
\end{lemma}
\begin{proof}
    The law $\proj(\QQ)$ is that of an orthogonal projection to a uniformly random subspace of $\RR^n$ of dimension $N^{\prime}$.
    In particular, the law of $\bP$ is orthogonally invariant, so by Lemma~\ref{lem:spherical-dom} and Lemma~\ref{lem:spherical-bound} we have
\begin{align*}
  \Ex_{\bP \sim \proj(\QQ)} p(\eta \bP)
  &\leq \Ex_{\bP \sim \proj(\QQ)}p^{\SS}(\eta\bP) \\
  &\leq \Ex_{\bP \sim \proj(\QQ)}\inf_{s > \eta / 2}\left\{ s - \frac{1 + \log(2)}{2} - \frac{1}{2n} \sum_{i = 1}^n \log\left(s - \frac{\lambda_i(\eta \bP)}{2}\right) \right\} \numberthis \label{eq:J-evals} \\
  &= \inf_{s > \eta/2}\left\{ s - \frac{1 + \log(2)}{2} - \frac{1 - \gamma^{-1}}{2}\log\left(s - \frac{\eta}{2}\right) - \frac{\gamma^{-1}}{2}\log(s) \right\} \numberthis \label{eq:before-s}
    \intertext{and, taking $s$ close to $\eta / 2$ and $\gamma$ sufficiently close to 1 depending on $\eta$, $s$, and $\delta$, we may ensure that}
  &\leq \frac{\eta}{2} - \frac{1 + \log(2)}{2} - \frac{1}{2}\log\left(\frac{\eta}{2}\right) + \frac{\delta}{2} \\
  &= \frac{\eta}{2} - \frac{1}{2} - \frac{1}{2}\log(\eta) + \frac{\delta}{2}. \numberthis
\end{align*}
The rest of the result then follows by a standard concentration inequality for $p(\eta \bP)$.
This may be derived by observing that we may sample $\bP = \bV\bV^{\top}$ where $\bV$ is drawn according to the Haar measure on the Stiefel manifold (the manifold of matrices of a given shape with orthonormal columns).
This Haar measure satisfies a log-Sobolev inequality, which implies the concentration of Lipschitz functions of $\bV$ (see the discussion following Theorem~2.4 in~\cite{Ledoux-2001-Conc} and Theorem~5.5 of \cite{Meckes-2019-RandomMatrixCompactGroups}).
Namely, if $L$ is the Lipschitz constant of the map $\bV \mapsto p(\eta \bV\bV^{\top})$ when the domain is endowed with the Frobenius norm, then
\begin{equation}
    \Prx_{\bP \sim \proj(\QQ)}\bigg[|p(\eta\bP) - \EE p(\eta\bP)| \geq t\bigg] \leq 2\exp\left(-\frac{Cnt^2}{L^2}\right)
\end{equation}
for an absolute constant $C$.
Towards computing the Lipschitz constant, we have
\begin{align*}
  |\bx^{\top}\bU\bU^{\top}\bx - \bx^{\top}\bV\bV^{\top}\bx|
  &= |\langle \bx\bx^{\top}, \bU\bU^{\top} - \bV\bV^{\top} \rangle| \\
  &\leq \|\bx\bx^{\top}\|_F\|\bU\bU^{\top} - \bV\bV^{\top}\|_F \\
  &= n\|\bU\bU^{\top} - \bV\bV^{\top}\|_F \\
  &= n\|\bU(\bU - \bV)^{\top} - (\bV - \bU)\bV^{\top}\|_F \\
  &\leq 2n \|\bU - \bV\|_F, \numberthis
\end{align*}
where we use the Cauchy-Schwarz inequality for the trace inner product of matrices and the triangle inequality and invariance under isometries of the Frobenius norm.
Therefore,
\begin{align*}
  p(\eta\bU\bU^{\top})
  &= \frac{1}{n}\log \left(\sum_{\bx \in \{\pm 1\}^n} \exp\left(\frac{\eta}{2}\bx^{\top}\bU\bU^{\top}\bx\right)\right) \\
  &\leq \frac{1}{n}\log \left(\sum_{\bx \in \{\pm 1\}^n} \exp\left(\frac{\eta}{2}\bx^{\top}\bV\bV^{\top}\bx + n \cdot \eta \|\bU - \bV\|_F \right)\right) \\
  &= p(\eta\bV\bV^{\top}) + \eta\|\bU - \bV\|_F, \numberthis
\end{align*}
and combining with symmetrical calculations shows
\begin{equation}
    |p(\eta\bU\bU^{\top}) - p(\eta\bV\bV^{\top})| \leq \eta\|\bU - \bV\|_F,
\end{equation}
so we may take $L = \eta$ for the Lipschitz constant.
In particular, with high probability $|p(\eta\bP) - \EE p(\eta\bP)| \leq \delta/2$, and the proof is complete.
\end{proof}

We remark that it is a simple computation to verify that, while this choice may be subtler in general, our choice of $s$ close to $\eta / 2$ in the proof is close to optimal in the parameter regime we are working in by taking a derivative in \eqref{eq:before-s}.
It is also clear that our choice of $\bJ$ to be a scaled low-rank projection matrix is optimal (among orthogonally invariant random matrices, which just leaves us to choose a distribution of eigenvalues) for our argument, which seeks to have the difference between the largest and smallest eigenvalues be large while making the quantity in \eqref{eq:J-evals} small.

We now proceed to the remainder of the main proof.
\begin{proof}[Proof of Theorem~\ref{thm:free-energy}]
    Let $\epsilon > 0$ be as in the statement of the Theorem.
    Let us take $\eta \colonequals 1 + \epsilon$.
    We will show below the crucial fact that, for any $\eta > 1$,
    \begin{equation}
        \label{eq:main-comparison}
        \bigg(\underbrace{\frac{\eta}{2} - \log(2) + \sup_{x \in [0, 1]} \left\{ H(x) + 2\eta x^2 - 2\eta x\right\}}_{\equalscolon \, c_{\PP}}\bigg) - \bigg(\underbrace{\frac{\eta}{2} - \frac{1}{2} - \frac{1}{2}\log(\eta)}_{\equalscolon \, c_{\QQ}}\bigg) \stackrel{?}{>} 0,
    \end{equation}
    but for the sake of exposition let us take this for granted for now.

    Let $\delta = \delta(\eta) \colonequals \frac{1}{6}(c_{\PP} - c_{\QQ}) > 0$.
    By Lemma~\ref{lem:QQ}, we may choose $\gamma = \gamma(\eta, \delta) > 1$ such that, when $\bP \sim \proj(\QQ)$, then with high probability
    \begin{equation}
        p(\eta\bP) \leq c_{\QQ} + \delta.
    \end{equation}
    By Lemma~\ref{lem:PP}, we may choose $\beta = \beta(\eta, \gamma) \in (-1, 0)$ such that, when $\bP\sim \proj(\PP)$, then with high probability
    \begin{equation}
        p(\eta\bP) \geq c_{\PP} - \delta.
    \end{equation}

    Suppose $\aapprox$ is an approximator satisfying the estimate in the Theorem for $\epsilon$ arbitrary as in the Theorem statement and the $\delta = \delta(\epsilon)$ we have specified above.
    We use $\aapprox$ to define a hypothesis test $f: (\RR^n)^N \to \{\texttt{p}, \texttt{q}\}$ as follows.
    Let $c \colonequals \frac{1}{2}(c_{\QQ} + c_{\PP})$.
    We then take
    \begin{equation}
        f(\by_1, \dots, \by_N) \colonequals \left\{ \begin{array}{ll} \texttt{p} & \text{if } \frac{1}{n}\, \aapprox(\eta \cdot \proj(\by_1, \dots, \by_N)) > c, \\ \texttt{q} & \text{otherwise}.\end{array} \right.
    \end{equation}

    By our assumption on $\aapprox$ and the above claims about $p(\eta\bP)$, we have with high probability when $(\by_1,\dots, \by_N) \sim \PP$ that
    \begin{equation}
        \frac{1}{n}\, \aapprox(\eta \cdot \proj(\by_1, \dots, \by_N)) \geq c_{\PP} - 2\delta = c + \delta,
    \end{equation}
    and with high probability when $(\by_1,\dots, \by_N) \sim \QQ$ that
    \begin{equation}
        \frac{1}{n}\, \aapprox(\eta \cdot \proj(\by_1, \dots, \by_N)) \leq c_{\QQ} + 2\delta = c - \delta,
    \end{equation}
    and therefore indeed $f$ will be a polynomial time hypothesis test that distinguishes $\PP$ from $\QQ$ with high probability, refuting Conjecture~\ref{conj:wishart}.

\begin{figure}
    \begin{center}
        \includegraphics[scale=0.9]{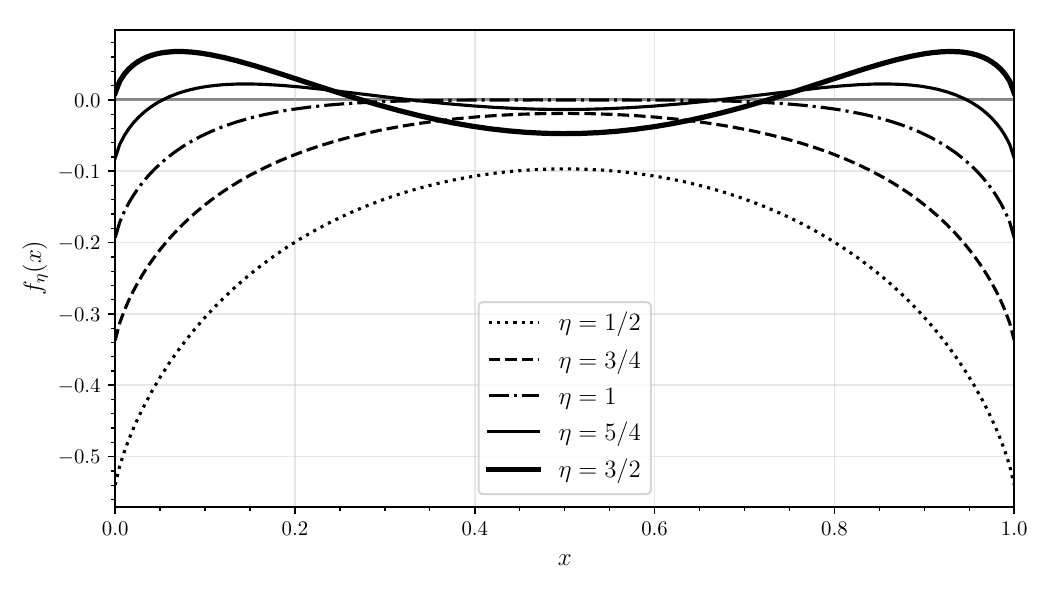}
    \end{center}
    \vspace{-2em}
    \caption{We plot $f_{\eta}(x)$ as appears at the end of the proof of Theorem~\ref{thm:free-energy} over $x \in [0, 1]$ for several values of $\eta > 0$.}
    \label{fig:f-eta}
\end{figure}

    It remains to show \eqref{eq:main-comparison}.
    Upon rearranging, this is equivalent to
    \begin{equation}
        c_{\PP} - c_{\QQ} = \sup_{x \in [0, 1]} \bigg\{ \underbrace{H(x) + 2\eta x^2 - 2\eta x + \frac{1}{2}\log(\eta) + \frac{1}{2} - \log(2)}_{\equalscolon f_{\eta}(x)} \bigg\} \stackrel{?}{>} 0.
    \end{equation}
    To verify this, we evaluate at $x_0 = x_0(\eta) \colonequals\frac{1}{2}(1 + \sqrt{1 - \frac{1}{\eta}})$ (this happens to be an inflection point of $f_{\eta}(x)$ between a local minimum and local maximum when $\eta > 1$; see Figure~\ref{fig:f-eta}).
    We find after simplifying that
    \begin{align*}
      f_{\eta}(x_0(\eta))
      &= \frac{1}{2}\log(\eta) - \log(2) + H\left(\frac{1}{2} + \frac{1}{2}\sqrt{1 - \frac{1}{\eta}}\right). \numberthis
    \end{align*}
    We have $f_1(x_0(1)) = 0$, so it suffices to show that, for all $\eta > 1$, $\frac{d}{d\eta}f_{\eta}(x_0(\eta)) > 0$.
    Calculating this derivative gives
    \begin{align*}
        \frac{d}{d\eta}f_{\eta}(x_0(\eta))
      &= \frac{1}{4\eta^2\sqrt{1 - \frac{1}{\eta}}}\left(2\eta\sqrt{1 - \frac{1}{\eta}} - \log\left(\frac{1 + \sqrt{1 - \frac{1}{\eta}}}{1 - \sqrt{1 - \frac{1}{\eta}}}\right)\right).\numberthis
    \end{align*}
    Introducing $\theta \colonequals \sqrt{1 - \frac{1}{\eta}}$, the term in parentheses above is
    \begin{equation}
        g(\theta) \colonequals \frac{2\theta}{1 - \theta^2} - \log\left(\frac{1 + \theta}{1 - \theta}\right),
    \end{equation}
    and it now suffices to show that this is positive on $\theta \in (0, 1)$.
    Again, $g(0) = 0$, and the derivative is
    \begin{equation}
        \frac{d}{d\theta}g(\theta) = \frac{4\theta^2}{(1 - \theta^2)^2},
    \end{equation}
    which is strictly positive on $\theta \in (0, 1)$, completing the proof.
\end{proof}
To illustrate the last calculation, we plot the curves $f_{\eta}(x)$ for various $\eta$ in Figure~\ref{fig:f-eta}, observing that indeed $\eta = 1$ is the transition point between there existing an $x \in [0, 1]$ such that $f_{\eta}(x) > 0$ and there existing no such $x$.

\begin{remark}
    One may ask whether the hypothesis test described above, used with an approximator of the free energy based on Glauber dynamics and some $\eta < 1$ but close to 1, is actually an effective testing algorithm when $\beta^2 > \gamma$, in which case it is computationally easy to distinguish $\PP$ from $\QQ$ (by thresholding the largest eigenvalue of the sample covariance, per the BBP transition).
    This is not the case in general: our analysis, in particular of the quantity $c_{\QQ}$, crucially relies on being able to choose $\gamma$ close to 1, so that $\bJ$ is a rescaled low-rank projection.
    It is intuitive that $\bJ$ having high rank would be an issue: in that case, the contribution of the planted hypercube vector close to the row space of $\bJ \sim \proj(\PP)$ is relatively small, and there is no longer a macroscopic difference in free energy as compared to $\bJ \sim \proj(\QQ)$.
\end{remark}

\begin{remark}
    We note that, as is visible in Figure~\ref{fig:f-eta}, one may show that for $\eta \leq 1$ the optimizer of the maximization in the lower bound on the planted model in Lemma~\ref{lem:PP} is $x^{\star} = \frac{1}{2}$, while once $\eta > 1$ there are superior values of $x$ other than $\frac{1}{2}$.
    In the former regime, we have $c_{\PP} = 0$, while in the latter we have $c_{\PP} > 0$.
    This may be interpreted as the planted hypercube vector not contributing macroscopically to the free energy of the planted model at high temperature $\eta \leq 1$, since $x = \frac{1}{2}$ corresponds to the contribution to the free energy of hypercube vectors uncorrelated (coinciding in half the coordinates) with the planted one and $c_{\PP}$ is a lower bound derived in Lemma~\ref{lem:PP} by considering only contributions due to correlations with the planted hypercube vector.
\end{remark}

\section*{Acknowledgments}
Thanks to Afonso Bandeira and Alex Wein for helpful comments on a draft, and to Andrea Montanari for introducing me to the argument described in Section~\ref{sec:related}.
Thanks also to the anonymous reviewers for several insightful comments.

\bibliographystyle{alpha}
\bibliography{main}

\end{document}